\documentclass[a4paper,USenglish,cleveref,autoref,thm-restate]{lipics-v2021}

\hideLIPIcs  

\title{Dimension Reduction for Curves: Simplified and Generalized}

\author{Matthijs Ebbens}{University of Cologne, Germany}{ymebbens@gmail.com}{}{Funded by the Deutsche Forschungsgemeinschaft
(DFG, German Research Foundation) -- Project Number 459420781.}

\author{Jie Lu}{University of Cologne, Germany}{jie.lu.duesseldorf@gmail.com}{}{}

\author{Alexander Munteanu}{TU Dortmund, Germany}{alexander.munteanu@tu-dortmund.de}{}{Funded by the Deutsche Forschungsgemeinschaft
(DFG, German Research Foundation) -- Project Number 535889065.}

\authorrunning{M. Ebbens, J. Lu and A. Munteanu} 

\Copyright{Matthijs Ebbens, Jie Lu and Alexander Munteanu} 



\category{} 

\relatedversion{} 

\nolinenumbers 

\EventEditors{Philip Bille, Seth Pettie, and Sabine Storandt}
\EventNoEds{3}
\EventLongTitle{34th Annual European Symposium on Algorithms (ESA 2026)}
\EventShortTitle{ESA 2026}
\EventAcronym{ESA}
\EventYear{2026}
\EventDate{August 31--September 4, 2026}
\EventLocation{L'Aquila, Italy}
\EventLogo{}
\SeriesVolume{388}
\ArticleNo{114}

\usepackage{mathtools}
\usepackage{amsmath,amsthm,amssymb}
\usepackage{xspace}

\newcommand{\R}{\mathbb{R}}

\newcommand{\mcT}{\mathcal{T}}

\newcommand{\eps}{\varepsilon}
\renewcommand{\epsilon}{\varepsilon}

\newcommand{\Frechet}{Fr\'echet\xspace}

\newcommand{\cF}{\operatorname{cF}}

\DeclareMathOperator*{\argmax}{arg\,max}

\newcommand{\REAL}{\ensuremath{\mathbb{R}}}

\allowdisplaybreaks

\begin{document}

\maketitle

\begin{abstract}
We revisit random projections for reducing the dimension of high-dimensional polygonal curves. Drawing from the toolbox of randomized linear algebra, we give a considerably simplified proof of the known $O(\varepsilon^{-2}\log(nm))$ bound on the target dimension of a random projection that preserves the continuous Fr\'echet distance of polygonal curves up to a factor $(1\pm\varepsilon)$.
Our proof is based on the concept of sparse oblivious subspace embeddings. While previous techniques were limited to the case of the Fr\'echet distance, our techniques are fairly general and extend to all possible distance measures that involve the maximum, a sum or an integral over Euclidean distances between pairs of points on both input curves. We define a generalized dissimilarity measure for curves that includes several popular measures such as Fr\'echet, $q$-DTW, Hausdorff, etc. as special cases and show that the same dimension reduction technique works for this generalized dissimilarity measure. Finally, we apply the same framework for dimension reduction to piecewise linear surfaces, after extending the distance measure suitably to such surfaces. 
\end{abstract}

\clearpage
\section{Introduction}
\label{sec:intro}

In this paper, we study dimension reduction for polygonal curves and build up a fairly general framework, which is widely independent of the dissimilarity measure that is used to compare them. Thus, in contrast to metric embeddings, our framework aims to embed \emph{the curves} into lower dimensional space, not a specific dissimilarity measure. We first motivate our work from a broader perspective before we return to polygonal curves and surfaces.

Dimension reduction is one of the most fundamental tasks in computational geometry and has many beneficial consequences. For instance, algorithms for a computational task may have a dependence on the dimension of their input objects, such as points, shapes, or polygonal curves. Reducing the dimension of these objects, say from $d$ to some $t\ll d$, allows to apply the same algorithm on the reduced data, but with the dependence on $d$ replaced by the significantly smaller value $t$. This comes at the cost of a small approximation error which is typically bounded within a factor of $(1\pm\eps)$. Another area where dimension reduction plays a crucial role is the construction of so-called \emph{coresets}~\cite{Badoiu02,AgarwalHV05}. Coresets aim to reduce the number of data objects rather than their dimension. Classic coreset constructions typically had a small, though polynomial dependence on $d$, depending on the problem for which they were developed. Sometimes exponential dependencies on the dimension were present or even unavoidable~\cite{AgarwalHV04,Chan04,MunteanuOP26}. More recently, for some problems such as $k$-means/median clustering, the dependence of the coreset size on the dimension could be completely eliminated~\cite{SohlerW18,FeldmanSS20}, which would not have been imaginable without (random) projection techniques such as the celebrated Johnson-Lindenstrauss (JL) Lemma~\cite{JL84}, PCA~\cite{KannanV09}, and their descendants. Extending these techniques to polygonal curves and surfaces is an intriguing open direction, as even the best coresets~\cite{Cohen-AddadD0SS25,ConradiKPR24} in this regime suffered from linear dimension dependence.

JL-type embeddings are key to dimension-reduced results in other areas as well: for instance in compressed sensing, one receives the promise that a $d$-dimensional dense and noisy signal is inherently $k$-sparse and the task is to reconstruct that signal with as few (random) measurements as possible. Here, the number of measurements corresponds to a reduced target dimension which is only $\Theta(k\log(d/k))$~\cite{Donoho06,CandesRT06,BaraniukDDW08}.

Dimension reduction tools have recently brought up dimension-free bounds for learning problems, where we would like to approximate some problem-specific 
quantity defined as the expectation taken over an unknown distribution of points or polygonal curves
~\cite{BucarelliLST23,Krivosija25}. The task is to sample a small (dimension-independent) number of data objects from the distribution such that the sample mean over the quantity is within a small error from its expectation. In contrast to coresets, the ground set is not necessarily finite and can only be accessed via i.i.d. samples from the unknown distribution, which complicates the problem and requires some natural assumptions to succeed~\cite{AlishahiP24,AlishahiMOP26}.

As a first step towards opening these avenues for problems on polygonal curves and surfaces, we propose a fairly general framework for embedding these objects in lower dimensional space, while preserving their `distance' as measured by popular (continuous) metrics, dissimilarity measures and divergences such as the \Frechet metric, dynamic time warping (DTW), Hausdorff distance, and their variants. In the following, we present our main results, techniques and related work, then prove our technical claims before concluding with a summary, certain limitations and possible extensions of our work.

\subsection{Results}\label{sec:results}
We present our main result, which is a general framework for dimension reduction of parameterized piecewise $\gamma$-dimensional linear surfaces. Notably, this includes polygonal curves as special case for $\gamma=1$.
For the notation and definitions of terms such as traversals, we refer to \Cref{sec:prelim}. For a more detailed explanation of the Theorem, we refer to \Cref{sec:surfaces}.
\begin{restatable}{theorem}{dimredsurfaces}
    \label{thm:dimredsurfaces}
    Let $\sigma_1,\ldots,\sigma_n:[0,1]^\gamma\rightarrow\R^d$ be $n$ piecewise linear surfaces with $m$ pieces. Let $\mcT$ be a set of admissible traversals between surfaces and let $\mu$ be a measure on $[0,1]^\gamma$. There exists a random linear mapping $f:\R^d\rightarrow\R^t$ with $t\in O(\varepsilon^{-2}(\gamma + \log(mn/\delta)))$ such that, with probability at least $1-\delta$, it holds for all $1\leq i< j\leq n$ that
    \[
        (1-\varepsilon)\,d(\sigma_i,\sigma_j)\leq d(f(\sigma_i),f(\sigma_j))\leq (1+\varepsilon)\,d(\sigma_i,\sigma_j),
    \]
    where the distance $d$ between surfaces $\sigma_i$ and $\sigma_j$ is defined as
    \[
        d(\sigma_i,\sigma_j)= \inf_{(\alpha,\beta)\in\mcT}\left(\int_{[0,1]^\gamma}\|\sigma_i(\alpha(r))-\sigma_j(\beta(r))\|^q_2 \,d\mu(r)\right)^{1/q}. 
    \]
\end{restatable}
The distance is highly general and includes the most popular continuous dissimilarity measures as special cases: \Frechet (for $q\rightarrow\infty$), DTW (for $q=1$), and its variants such as weak \Frechet, $q$-DTW for $q\in [0,\infty)$, etc. 
While Hausdorff distance is not directly subsumed by the generalized dissimilarity measure, an extension is immediate from our results.

Discrete variants are also included by choosing the measure $\mu$ of integration to be a counting measure.
Discrete traversals for polygonal curves are well-known, but no such notion is known for surfaces; in particular there is no natural ordering of the vertices. We thus propose a set-valued generalization of discrete traversals to piecewise linear surfaces and prove that our generalization is consistent with the standard definition in the case of polygonal curves.

\subsection{Related work}\label{sec:relatedwork}
\subparagraph{Random projections for curves under \Frechet distance}
Early work \cite{DriemelK18} studied probabilistic embeddings of the \Frechet distance between pairs of polygonal curves of complexity bounded by at most $m$ onto random lines, resulting in large multiplicative distortions of order $\Theta(m)$.
Allowing larger embedding dimensions, we note that preserving the \emph{discrete} \Frechet and DTW distances of $n$ curves is straightforward from the JL lemma: it suffices to preserve all $m^2$ pairwise distances of vertices between $\binom{n}{2}$ pairs of curves, because every discrete traversal that we consider in a maximization or sum, consists of a subset of these distances. The embedding dimension is thus $O(\eps^{-2}\log(nm))$.

Maybe surprisingly, the same embedding dimension suffices to preserve the continuous case as well. A first step was made in \cite{MeintrupMR19} who achieved mixed relative and additive errors using the same embedding of all vertices. A relative $(1\pm\eps)$ error guarantee was proven more recently in \cite{PsarrosR23} using a complicated argument based on logical predicates to decide whether the \Frechet distance is within some value $\delta$. These predicates check the aforementioned pairwise distances between vertices and additionally include distances to a couple of anchor points along edges to avoid certain complications with very long edges that caused additive errors in the previous analysis of~\cite{MeintrupMR19}. The set of all vertices and anchor points builds a discrete `skeleton' for the \Frechet distance. The entire set of skeleton points is added to the embedding and since their number remains polynomially bounded in both $n$ and $m$, this does not increase the asymptotical $O(\eps^{-2}\log(nm))$ target dimension of the embedding. The journal version~\cite[Section 1.2.1]{PsarrosR25} included an informal note on oblivious subspace embeddings similar to our work. They did not provide a proof and only considered the \Frechet distance, whereas our techniques extend and generalize to other distance measures and to surfaces.

\subparagraph{Oblivious subspace embeddings and sparsity}
The finite set embedding guarantee of the celebrated JL lemma \cite{JL84} was first extended to preserve $k$-dimensional linear subspaces in the compressed sensing literature \cite{Donoho06,CandesRT06,BaraniukDDW08}, where we get the promise that some signal is $k$-sparse. We note that this only requires an embedding of axis aligned subspaces, whose bases are subsets comprising $k$ elements of the standard basis in $\R^d$. This is often referred to as the restricted isometry property (RIP). The next step towards arbitrary $k$-dimensional subspaces represented in arbitrary bases\footnote{While an explicit basis is useful for conducting formal proofs, the embeddings are \emph{oblivious} and do not require an explicit representation to work; the existence of a subspace is sufficient to embed it.} was made by Sarl\'os \cite{Sarlos06}, who paved the way for the fastest approximation algorithms for numerical linear algebra \cite{ClarksonW09,Woodruff14} with a plethora of applications in Machine Learning, cf.~\cite{DrineasMMW12,Munteanu22}. Finally, the step back to $k$-dimensional subspaces with a $k$-sparse support in an arbitrary, not necessarily axis-aligned basis, was recently studied in the context of sparse regression \cite{MaiMM0SW23}. Our contributions show how to exploit these works for dimension reduction for polygonal curves and piecewise linear surfaces.

\subsection{Technical overview}\label{sec:techniques}
\subparagraph{Warm up: dimension reduction for the \Frechet distance}
Our investigations start with the continuous \Frechet distance as a warm up. We follow a different approach than recent work \cite{PsarrosR23,PsarrosR25} detailed in \Cref{sec:relatedwork} on related work. Instead of discretizing the polygonal curves directly to preserve the \Frechet distance and applying the mapping $f$ asserted by the JL lemma to the discrete set, we interpret the Euclidean distances between \emph{any} points $p,q$ on either curve as norms of their connecting vectors $v=p-q$. Our problem reduces to preserving the norm of all possible vectors $v$. Now, the multiplicative approximation guarantee that we get from the JL lemma is
\[
    (1-\eps) \|v\|_2 \leq \|f(v)\|_2 \leq (1+\eps) \|v\|_2\,
\]
which is equivalent to a ($1$-dimensional) so-called subspace embedding, which guarantees that
\[
    \forall x\in \R \colon (1-\eps) \|vx\|_2 \leq \|f(vx)\|_2 \leq (1+\eps) \|vx\|_2\,.
\]
The latter guarantee can be generalized to preserve arbitrary $k$-dimensional subspaces based on the JL lemma with embedding dimension $t\in O(\eps^{-2}(k+\log(1/\delta)))$ where $\delta$ denotes the failure probability~\cite{Sarlos06,Woodruff14}. Let $V\in\R^{d\times k}$ be a matrix whose columns build a basis for a $k$ dimensional linear subspace of $\R^d$. Then choosing a properly rescaled i.i.d. Gaussian matrix $S\in\R^{t\times d}$ yields the guarantee that with probability at least $1-\delta$ it satisfies~\cite{Sarlos06,Woodruff14}
\[
    \forall x\in \R^k \colon (1-\eps) \|Vx\|_2 \leq \|SVx\|_2 \leq (1+\eps) \|Vx\|_2\,.
\]
Now, we reduce the task of preserving the \Frechet distance via Gaussian matrices to preserving a \emph{collection} of $k$-dimensional subspace embeddings for some $k\in O(1)$ as follows. We consider two polygonal curves $\pi$ and $\tau$ in $\R^d$. We observe that any pair of points $p\in \pi$ and $q \in \tau$, not necessarily vertices, can be expressed as a convex combination of their two neighboring vertices. The vector $v = p-q$ may thus be expressed as a linear combination of $4$ vertices. In fact \emph{any} vector determining the \Frechet distance lies in some $4$-dimensional linear subspace. Preserving the collection of all such `$4$-sparse' subspaces will thus suffice.

Now, consider a set of $n$ curves with vertices in $\R^d$ of complexity at most $m$. We let $A\in\R^{d \times nm}$ be the matrix whose columns comprise all up to $nm$ vertices of all input curves. Then any norm we are interested in, takes the form $\|Ax\|_2$, where we get the promise that at most $k=4$ many coordinates of $x$ are non-zero at a time. We can now construct a $k$-sparse subspace embedding~\cite{MaiMM0SW23} that satisfies
\[
    \forall x\in \Psi_k \colon (1-\eps) \|Ax\|_2 \leq \|SAx\|_2 \leq (1+\eps) \|Ax\|_2\,,
\]
where $\Psi_k = \{x \in \R^{nm} \mid \|x\|_0 \leq k\}$.
Recall that one single $k$-dimensional subspace embedding can be achieved with target embedding dimension $O(\eps^{-2}(k+\log(1/\delta)))$. A simple combinatorial calculation yields that there exist at most $\binom{nm}{k}\leq (\frac{enm}{k})^k$ choices. Using the union bound it thus suffices to choose a target dimension of $t\in O(\eps^{-2}k\log(nm/(\delta k)))\subseteq O(\eps^{-2}\log(nm/\delta))$, since $k=4$ is an absolute constant.

Finally, by an exchange argument between the optimal traversals w.r.t. the original and w.r.t the embedded curves, and using the subspace embedding to switch between original and low-dimensional curves, we show that the \Frechet distance is preserved up to $(1\pm\eps)$ factors.

\subparagraph{Extension beyond \Frechet}
We observe that the dimension reduction and the final exchange argument are much more general. Indeed, preserving an infimum over integrals for the continuous DTW in the same way seems even more direct than preserving the inf/sup combination for the \Frechet distance. Regardless of the difficulty of the exchange argument, by embedding the collection of $4$-sparse subspaces, we have actually embedded all Euclidean distances that can ever appear in a plethora of dissimilarity measures that are variants of \Frechet and DTW. Even others are included, such as Hausdorff distance, which has a different formal definition but is again based on a vector space equipped with the Euclidean norm.

\subparagraph{Piecewise linear surfaces}
Our techniques can be extended to handle piecewise linear surfaces of dimension $\gamma$, where polygonal curves constitute the special case $\gamma=1$. Our argument works as follows. Consider two piecewise linear surfaces. Let $p,q$ be points on some linear piece of either of the two surfaces. By Carath\'eodory's theorem there exists a convex combination of vertices $v_i$ such that $p=\sum_{i=1}^{\gamma+1} \lambda_i v_i$. The same holds for $q$, so the difference vector $p-q$ can be expressed as a linear combination of $k=2(\gamma+1)$ vertices. To preserve the norms of all such linear combinations, we can again apply the $k$-sparse subspace embedding of~\cite{MaiMM0SW23} with target embedding dimension $O(\eps^{-2}(\gamma\log(nm/(\delta \gamma))))$. In our technical proofs in \Cref{sec:surfaces}, we exploit more geometric structure of our problem to reduce this to $O(\eps^{-2}(\gamma+\log(nm/\delta)))$.

Our generalized dissimilarity measure again serves as a general framework for the surface variants of continuous $q$-DTW, including the limiting cases $q=1$ (continuous DTW), and $q\rightarrow \infty$ (\Frechet). Hausdorff distance can also be handled in a similar way as for curves. Continuous reparameterizations between two surfaces can be handled using homeomorphisms.

\subparagraph{Discrete dissimilarity measures}
Dimension reduction for discrete dissimilarity measures for curves is folklore, and can be handled by a straightforward application of the JL lemma to all pairwise distances between vertices of two curves. It is subsumed in our framework, and can be implemented by using a counting measure instead of a continuous measure $\mu$ for integration. For surfaces, however, it is unclear how a discrete traversal should be defined. We thus handle this issue by taking a detour over a continuous homeomorphism $h$, mapping vertices to one another whenever $h$ or $h^{-1}$ maps a vertex to the Voronoi cell of the other. 
For a fixed choice $h\colon A \rightarrow B$, we consider the Voronoi cells $V(a_i), V(b_i)$ of the vertices $a_1,\ldots,a_m$ respectively $b_1,\ldots,b_m$ on each surface. Now, whenever $h(a_i)\in V(b_j)$, we add the pair $(a_i,b_j)\in T$ without duplicates. Similarly, we also consider the inverse homeomorphism and add $(a_j,b_i)\in T$ without duplicates whenever $h^{-1}(b_i)\in V(a_j)$.
We prove that a discrete traversal constructed in this way is consistent with the standard definition of discrete traversals for polygonal curves. Ranging over all homeomorphisms, we thus enumerate the set of all possible discrete traversals $T\in\mathcal T$ between surfaces or curves.

\section{Preliminaries}\label{sec:prelim}
\subparagraph{Notation}
We set $[n]\coloneqq \{1,\ldots,n\}$. For a vector $v\in\R^d$, and $p\in(0,\infty)$, we denote its $\ell_p$-norm by $\|v\|_p = \left(\sum_{i=1}^d |v_i|^p \right)^{1/p}$. We remark that the cases $0 < p < 1$ only define quasi-norms. Other special cases are $\|v\|_\infty\coloneqq\underset{p\rightarrow\infty}{\lim} \|v\|_p = \max_{i\in[d]}|v_i|$, and $\|v\|_0 = |\{ i\in [d]\mid v_i \neq 0\}|$. We denote by $\Psi_k \coloneqq \{x\in \R^d\mid \|x\|_0 \leq k\}$ the set of $k$-sparse vectors in $\R^d$.

\subparagraph{Polygonal curves and surfaces}
A \emph{curve} in Euclidean space $\R^d$, is a continuous function $\sigma:[0,1]\rightarrow \REAL^d$. A \emph{polygonal curve} is a curve such that there exist a finite number $m\in \mathbb{N}$ of values $0=s_1\leq s_2\leq \ldots \leq s_m=1$, with $v_i=\sigma(s_i)$ which we call \emph{vertices}. We say that $\sigma$ has \emph{complexity} $m$. Further, for each $i\in\lbrace 1,\ldots,m-1\rbrace$ the two consecutive vertices $\sigma(s_i)$ and $\sigma(s_{i+1})$ are connected by an affine straight-line segment, i.e., it holds that
\[
\sigma(s_i+\varsigma)= \left( 1-\frac{\varsigma}{s_{i+1}-s_i}\right) \cdot \sigma(s_i) + \frac{\varsigma}{s_{i+1}-s_i} \cdot \sigma(s_{i+1})\,,
\]
for all $\varsigma\in [0, s_{i+1}-s_i]$.
For simplicity of notation, we can fully characterize a polygonal curve $\sigma$ by defining the sequence of its vertices $\sigma=(v_1,\ldots, v_m)$.

Analogously, a piecewise linear \emph{surface} in $\R^d$ is a continuous function $\sigma:[0,1]^\gamma \rightarrow \REAL^d$ whose image consist of $m$ linear pieces which are polygons/polyhedrons bounded by a total of $N$ vertices $v_i\in\R^d$, such that every piece lies in a $\gamma$-dimensional affine subspace of $\R^d$.

\subparagraph{Traversals and dissimilarity measures}
We work with discrete and continuous `transformations'.
One way to define them is the notion of \emph{traversals}.
Let $\sigma=( u_1,\ldots,u_{m'})$ and $\tau=( v_1,\ldots,v_{m''} )$ be two curves of complexities $m'$ and $m''$ respectively.
Traversals are usually defined by inductively constructed sequences of point pairs $(p,q)$ such that $p\in \sigma$ and $q\in \tau$. It will be helpful for the sake of generalization to surfaces, to work with equivalent (unordered) sets of point pairs instead. We prove their equivalence in \Cref{app:traversal_equiv}.

Then, a discrete \emph{traversal} of $\sigma$ and $\tau$ is a set $T$ that consists of pairs of vertices $(p,q)$ with $p\in\sigma$ and $q\in\tau$, such that 

\begin{enumerate}
    \item $(u_1,v_1) \in T \wedge (u_{m'},v_{m''})\in T$,
    \item $\forall i<m', j<m''\colon (u_i,v_j)\in T_h \Rightarrow (u_i,v_{j+1})\in T_h \vee (u_{i+1},v_j)\in T_h\vee (u_{i+1},v_{j+1})\in T_h$,
    \item $(u_i,v_{j+1})\in T_h \Rightarrow (u_{i+1},v_j)\notin T_h$.
\end{enumerate}

We remark that the property of discrete traversals of not being allowed to backtrack, as common in the literature, is encapsulated by the third criterion. However, the equivalence between defining traversals as sequences of point pairs or as sets of point pairs holds even if traversals are allowed to backtrack: in this case, the set definition can be easily adapted. The notion of discrete traversals can be extended to the continuous case by specifying pairs of `reparameterizations', which are functions $\alpha,\beta\colon [0,1]^\gamma \rightarrow [0,1]^\gamma$. Then we let the induced traversal be the set $T=\{(\sigma(\alpha(r)),\tau(\beta(r))) \mid r\in [0,1]^\gamma \}$. Sometimes it is appropriate to impose further restrictions on $\alpha,\beta$ such as injectivity, surjectivity, continuity, monotonicity, as required, and denote $\mathcal T\subseteq\{(\alpha,\beta)\mid \alpha,\beta\colon [0,1]^\gamma \rightarrow [0,1]^\gamma\}$ the corresponding subset of \emph{admissible} traversals.

Given two curves $\sigma,\tau$, we now define a general dissimilarity measure $d(\sigma,\tau)$ between surfaces including curves as the special case where $\gamma=1$. We let
\[
    d(\sigma,\tau)\coloneqq \inf_{(\alpha,\beta)\in\mcT}\left(\int_{[0,1]^\gamma}\|\sigma(\alpha(r))-\tau(\beta(r))\|^q_2 \,d\mu(r)\right)^{1/q}\,. 
\]
This dissimilarity measure subsumes a plethora of popular distances. For instance, the \Frechet distance, $d_{\cF}(\sigma,\tau) = \inf_{(\alpha,\beta)}\max_{r\in[0,1]}\|\sigma(\alpha(r))-\tau(\beta(r))\|_2$, is implemented by the following choices: we take the limit of $q\rightarrow \infty$, in which case the $\mathcal{L}^q$ integration becomes the supremum, and in case of compact shapes, even simply the maximum. A uniform measure $\mu$ over $[0,1]^\gamma$ is chosen. For curves, the set of admissible traversals are restricted to be pairs of continuous, non-decreasing surjections (in which case the infimum is attained). An alternative that generalizes to surfaces is to take homeomorphisms. The infimum is then attained in the limit of a sequence of homeomorphisms.

For continuous DTW, we set $q=1$.
The other items are similar to the \Frechet distance except that for technical reasons, we would need to extend the interval $[0,1]$ to $[0,p+q]$, where $p,q$ equal the arc-lengths of $\sigma,\tau$ respectively, see \cite{BuchinNW22} for details. We stress that these technical details do not matter for our dimension reduction.
Discrete variants can be realized by choosing a counting measure for $\mu$, which transforms the integral to a finite sum.

\subparagraph{Subspace embeddings} The notion of $\eps$-subspace embeddings \cite{Sarlos06} is central to our work. The following version is due to \cite{Woodruff14} and is slightly modified to fit our notation and setting.
\begin{theorem}[$\eps$-Subspace Embedding, Theorem 2.3 of \cite{Woodruff14}]
\label{thm:SE}
	Let $A\in \R^{d \times k}$. 
    There exists a distribution over random matrices $S\in\R^{t\times d}$ with $t\in O(\eps^{-2} (k + \log(1/\delta)))$
    such that it holds with probability $1-\delta$ that
	\[
        \forall x\in\R^d: (1-\eps) \|Ax\|_2 \leq \|SAx\|_2 \leq (1+\eps) \|Ax\|_2 .
    \]
\end{theorem}
We note that the matrix $S$ can be chosen such that every $S_{i,j}\overset{i.i.d.}{\sim}\mathcal N(0,1/t)$. The same result can be achieved by rescaled random sign (Rademacher) matrices, which is more convenient in streaming and other memory constrained settings \cite{ClarksonW09}. 
{A crucial property is that the construction of $S$ is \emph{oblivious} to the matrix $A$. That means we can apply it to preserve subspaces of dimension $k$ for which we do not have an explicit representation or whose basis column vectors are hidden in a superset of columns. The existence of a subspace suffices for the guarantee to apply, cf.~\cite{Woodruff14}.}
The idea behind the subspace embedding construction of \cite{Sarlos06} is that the unit ball in the $k$-dimensional subspace of $\R^d$ spanned by the columns of $A$ can be covered by an $\eps$-net of size $(1+2/\eps)^k$~\cite{Pis99}. By a simple vector chaining argument covered in \cite[Theorem 2.1]{Woodruff14}, the net can be constructed with an absolute constant $\eps_0:=1/2$ instead of $\eps$ to still hold for every $\eps\leq \eps_0$. This yields a net $\mathcal N$ of size $|\mathcal N| = C^k$ for an absolute constant $C$. Applying the JL lemma with a union bound over this net yields a $(1\pm\eps)$ approximation to the norm of every point in the net within target dimension $O(\eps^{-2}\log(|\mathcal N|/\delta))\subseteq O(\eps^{-2} (k \log(C)+\log(1/\delta)))\subseteq O(\eps^{-2} (k + \log(1/\delta)))$. To complete the argument, every point in the ball has a point in the net at distance at most $\eps$ to which it can be related to achieve $(1\pm O(\eps))$ error. By linearity of the embedding and the multiplicative error guarantee, the argument extends from the unit ball to arbitrary radius by scaling.

We will exploit $\eps$-subspace embeddings for \emph{collections} of $k$-dimensional subspaces.
To this end we use the following theorem of \cite{MaiMM0SW23} slightly modified to fit our notation and setting.
\begin{theorem}[$k$-Sparse $\eps$-Subspace Embedding, Theorem 7 of \cite{MaiMM0SW23}]
\label{thm:kSE}
	Let $A\in \R^{d \times N}$.
    There exists a distribution over random matrices $S\in\R^{t\times d}$ with $t\in O(\eps^{-2}k\log(N/(\delta k)))$
    such that it holds with probability $1-\delta$ that
	\[
        \forall x\in\Psi_k: (1-\eps) \|Ax\|_2 \leq \|SAx\|_2 \leq (1+\eps) \|Ax\|_2 .
    \]
\end{theorem}
The idea is that there are $\binom{N}{k}\leq ({eN}/{k})^k$ choices of $k$ non-zero coordinates out of $N$. Every choice corresponds to a linear subspace of dimension at most $k$ that is spanned solely by the columns associated with non-zero entries in $x$. Using a union bound over all these choices, \Cref{thm:SE} applies with target dimension
$t\in O(\eps^{-2}k\log({N}/{(\delta k)}))$ 
and yields $(1\pm\eps)$ distortion for the Euclidean norm of all vectors in the collection of $k$-sparse subspaces.

\section{\Frechet distance between curves}\label{sec:frechet}
\subsection{Dimension reduction}\label{sec:dimredfrechet}
In this section, we will prove the following proposition, which can be seen as the restriction of \Cref{thm:dimredsurfaces} to the continuous \Frechet distance. This result is already known~\cite{PsarrosR23,PsarrosR25} as discussed in \Cref{sec:relatedwork}. However, our approach is different, in particular considerably simpler and shorter and it serves as a warm-up for proving \Cref{thm:dimredsurfaces} in full generality.

\begin{proposition}
    Let $\sigma_1,\ldots,\sigma_n:[0,1]\rightarrow\R^d$ be $n$ piecewise linear curves with $m$ vertices. There exists a random linear mapping $f:\R^d\rightarrow\R^t$ with $t \in O(\varepsilon^{-2}\log(mn/\delta))$ such that, with probability at least $1-\delta$ it holds for all $1\leq i<j\leq n$ that
    \[
        (1-\varepsilon)\,d_{\cF}(\sigma_i,\sigma_j)\leq d_{\cF}(f(\sigma_i),f(\sigma_j))\leq (1+\varepsilon)\,d_{\cF}(\sigma_i,\sigma_j), 
    \]
    where $d_{\cF}(\sigma_i,\sigma_j) = \inf_{(\alpha,\beta)}\max_{r\in[0,1]}\|\sigma_i(\alpha(r))-\sigma_j(\beta(r))\|_2$ denotes the continuous \Frechet distance between curves.
\end{proposition}

\begin{proof}
    The proof consists of two parts: first, we show that there exists a random linear mapping $S:\R^d\rightarrow\R^t$ with $t\in O(\varepsilon^{-2}\log(mn/\delta))$ such that 
    \[
        (1-\varepsilon)\|p-q\|_2\leq \|S(p)-S(q)\|_2\leq(1+\varepsilon)\|p-q\|_2 
    \]
    for every pair of points $p\in\sigma_i,q\in\sigma_j$ on every pair of curves $\sigma_i,\sigma_j$. We emphasize that $p$ and $q$ do not have to be vertices of $\sigma_i$ and $\sigma_j$. Then, we show that this leads to the desired approximation for the continuous \Frechet distance between $\sigma_i$ and $\sigma_j$.
    
    Consider any pair of curves $\sigma_i,\sigma_j$ and any pair of points $p\in \sigma_i$ and $q \in \sigma_j$. Denote the vertices of $\sigma_i$ and $\sigma_j$ by $p_1,\ldots,p_m$ and $q_1,\ldots,q_m$, respectively.
    Observe that $p$ and $q$ can be expressed as convex combinations of their neighboring vertices, i.e., there exist $\lambda_p,\lambda_q\in [0,1]$ and indices $h,\ell\in[m-1]$ such that 
    $p=(1-\lambda_p) p_h + \lambda_p p_{h+1}$ and $q=(1-\lambda_q) q_\ell + \lambda_q q_{\ell+1}$. Then, $
    p-q = (1-\lambda_p) p_h + \lambda_p p_{h+1} - (1-\lambda_q) q_\ell - \lambda_q q_{\ell+1}$
    is a linear combination of $4$ vertices. 
    
    Hence, to preserve $\|p-q\|_2$ up to a factor $(1\pm\varepsilon)$ for all $p\in\sigma_i,q\in\sigma_j$, it suffices to preserve the norm of all linear combinations of $4$ vertices of $\sigma_i$ and $\sigma_j$. Let $A\in\R^{d\times mn}$ be the matrix whose columns comprise all vertices of all curves $\sigma_1,\ldots,\sigma_n$. Then, any linear combination of $4$ vertices of $\sigma_i$ and $\sigma_j$ can be written as $Ax$ for some $4$-sparse vector $x\in\R^{mn}$. By setting $k=4$ and applying \Cref{thm:kSE} with $N=nm$, there exists a random linear mapping $S:\R^{d}\rightarrow\R^t$ with $t\in O(\varepsilon^{-2}\log(mn/\delta))$ such that 
        $(1-\eps) \|Ax\|_2 \leq \|SAx\|_2 \leq (1+\eps) \|Ax\|_2$ 
    holds for all $4$-sparse $x$. It follows for every pair of points $p\in\sigma_i,q\in\sigma_j$ that 
    \[
        (1-\varepsilon)\|p-q\|_2\leq \|S(p)-S(q)\|_2\leq(1+\varepsilon)\|p-q\|_2 
    \]
    holds for every pair of curves $\sigma_i,\sigma_j$, which is what we wanted to show.
    
    Now, to show that $S$ preserves the continuous \Frechet distance between any pair of curves $\sigma_i$ and $\sigma_j$, let $(\alpha^*,\beta^*)$ and $(\alpha^S,\beta^S)$ be reparameterizations realizing $d_{\cF}(\sigma_i,\sigma_j)$ and $d_{\cF}(S(\sigma_i),S(\sigma_j))$ respectively. Let $r'\in\argmax_{r\in[0,1]}\|S(\sigma_i(\alpha^*(r)))-S(\sigma_j(\beta^*(r)))\|_2$. Then,
    \begin{align*}
        d_{\cF}(S(\sigma_i),S(\sigma_j))&=\inf_{(\alpha,\beta)}\max_{r\in[0,1]}\|S(\sigma_i(\alpha(r)))-S(\sigma_j(\beta(r)))\|_2 \\
        &\le \max_{r\in[0,1]}\|S(\sigma_i(\alpha^*(r)))-S(\sigma_j(\beta^*(r)))\|_2 \\
        &= \|S(\sigma_i(\alpha^*(r')))-S(\sigma_j(\beta^*(r')))\|_2 \\
        &\le (1+\varepsilon) \,\|\sigma_i(\alpha^*(r'))-\sigma_j(\beta^*(r'))\|_2 \\
        &\le (1+\varepsilon)\max_{r\in[0,1]}\|\sigma_i(\alpha^*(r))-\sigma_j(\beta^*(r))\|_2 \\
        &=(1+\varepsilon)\inf_{(\alpha,\beta)}\max_{r\in[0,1]}\|\sigma_i(\alpha(r))-\sigma_j(\beta(r))\|_2 \\
        &=(1+\varepsilon)\,d_{\cF}(\sigma_i,\sigma_j)\,,
    \end{align*}
    which proves the dilation bound.
    Next, let $r''\in\argmax_{r\in[0,1]}\|\sigma_i(\alpha^S(r))-\sigma_j(\beta^S(r))\|_2$. Then it holds that
    \begin{align*}
        d_{\cF}(S(\sigma_i),S(\sigma_j))&=\inf_{(\alpha,\beta)}\max_{r\in[0,1]}\|S(\sigma_i(\alpha(r)))-S(\sigma_j(\beta(r)))\|_2\\
        &= \max_{r\in[0,1]}\|S(\sigma_i(\alpha^S(r)))-S(\sigma_j(\beta^S(r)))\|_2 \\
        &\ge \|S(\sigma_i(\alpha^S(r'')))-S(\sigma_j(\beta^S(r'')))\|_2 \\
        &\ge (1-\varepsilon)\,\|\sigma_i(\alpha^S(r''))-\sigma_j(\beta^S(r''))\|_2 \\
        &=(1-\varepsilon)\max_{r\in[0,1]}\|\sigma_i(\alpha^S(r))-\sigma_j(\beta^S(r))\|_2 \\
        &\ge (1-\varepsilon)\inf_{(\alpha,\beta)}\max_{r\in[0,1]}\|\sigma_i(\alpha(r))-\sigma_j(\beta(r))\|_2 \\
        &=(1-\varepsilon)\,d_{\cF}(\sigma_i,\sigma_j)\,,
    \end{align*}
    which proves the contraction bound. This finishes the proof.
\end{proof}

\section{Distance measures between piecewise linear surfaces}\label{sec:surfaces}
In this section, we first extend our framework for dimension reduction to distances between surfaces. Then, we move to discrete distance measures by giving a consistent definition of discrete traversals for surfaces.

\subsection{Dimension reduction}\label{sec:dimredsurfaces}
In this subsection, we prove \Cref{thm:dimredsurfaces}.
The outline of our proof is very similar to the special case of the \Frechet distance treated in \Cref{sec:frechet}. However, to achieve the target dimension bound of $O(\eps^{-2}(\gamma+\log(nm/\delta)))$, it is not sufficient to apply the sparse subspace embedding given in \Cref{thm:kSE}. Instead, we exploit the geometric structure of the problem to avoid considering the same $\gamma$-dimensional subspaces multiple times and other overcounting issues that did not matter in the case of polygonal curves, where $\gamma\in O(1)$. Moreover, standard reparameterizations are insufficient to handle the case of surfaces. We thus work with homeomorphisms, which require to handle the infimum via sequences of homeomorphisms, because, unlike reparameterizations, there are cases where the infimum is not attained.

\dimredsurfaces*
\begin{proof}
    The proof again consists of two parts: first, we show that there exists a random linear mapping $S:\R^d\rightarrow\R^t$ with $t\in O(\eps^{-2}(\gamma+\log(nm/\delta)))$ such that 
    \[
        (1-\varepsilon)\|p-q\|_2\leq \|S(p)-S(q)\|_2\leq(1+\varepsilon)\|p-q\|_2 
    \]
    for every pair of points $p\in\sigma_i,q\in\sigma_j$ on every pair of distinct piecewise linear surfaces $\sigma_i,\sigma_j$. Note that $p$ and $q$ do not have to be vertices of $\sigma_i$ and $\sigma_j$. Then, we show that this leads to the desired approximation for the generalized distance measure between $\sigma_i$ and $\sigma_j$.
    
    Consider any pair of piecewise linear surfaces $\sigma_i,\sigma_j$ and any pair of points $p\in \sigma_i$ and $q \in \sigma_j$. 
    Observe that $p$ lies in some linear piece $P$ of $\sigma_i$ that may be non-convex or contain holes. Consider the convex hull of $P$ and a triangulation thereof into simplices whose vertices are subsets of $\gamma+1$ vertices of $\sigma_i$.

    Now, $p$ lies in one of these simplices and by Carath\'eodory's theorem, $p$ can be expressed as a convex combination of the vertices of this simplex. This means that there exist $\gamma+1$ vertices $v_1,\ldots,v_{\gamma+1}$ of $\sigma_i$ and coefficients $\lambda \in \Delta_{\gamma+1}$ such that $p=\sum_{i=1}^{\gamma+1} \lambda_i v_i$, where $\Delta_{\gamma+1} \coloneqq \{\lambda\in \R^{\gamma+1} \mid \sum_{i=1}^{\gamma+1}\lambda_i = 1, \forall i\in[\gamma+1]\colon \lambda_i \geq 0\}$ is the probability simplex in $\gamma+1$ dimensions. The same argumentation holds for $q$ and thus we have $\lambda, \eta \in \Delta_{\gamma+1}$, and vertices $v_1,\ldots,v_{\gamma+1}$ of $\sigma_i$ and vertices $w_1,\ldots,w_{\gamma+1}$ of $\sigma_j$ such that
    \[
        \|p-q\|_2 = \left\|\sum_{i=1}^{\gamma+1} \lambda_i v_i - \sum_{i=1}^{\gamma+1} \eta_i w_i\right\|_2 \,.
    \]
    Hence, the difference vector can be expressed as a linear combination of $2(\gamma+1)$ vertices.

    At this point, we could again apply the sparse subspace embedding, given in \Cref{thm:kSE}, to the matrix $A\in \R^{d\times N}$ carrying the collection of all $N\geq nm(\gamma+1)$ vertices as its columns, where the sparsity bound parameter is $k=2(\gamma+1)$. This would result in an embedding dimension $t\in O(\eps^{-2}(\gamma\log(N/(\delta \gamma))))$. In particular, $N$, the number of vertices could be much larger than $nm(\gamma+1)$, and additionally we wish to reduce the factor $\gamma$ outside the logarithm to an additive $\gamma$. Our ultimate goal is to reduce from $d$ to $t\in O(\eps^{-2}(\gamma+\log(nm/\delta)))$ dimensions.
    
    To this end, we need to eliminate some overcounting in our previous argument. Observe that every linear piece lies in an affine subspace which requires a basis of only $\gamma+1$ columns to be preserved. In particular \emph{any} non-degenerate set of vertices of one single simplex will suffice as a representative. If we consider any other simplex within the same piece, it will span the same subspace. The only change is that instead of convex combinations $\lambda,\eta\in \Delta_{\gamma+1}$ of these vertices, we need to consider linear combinations $\lambda,\eta \in \R^{\gamma+1}$ to express $p$ resp. $q$. The resulting distance vector $p-q$ is again a linear combination of at most $k=2(\gamma+1)$ vertices, half of which are located on a linear piece on either surface $\sigma_i$ resp. $\sigma_j$.

    By our refined argument, the number of $k$-dimensional linear subspaces we need to consider is thus bounded by $\binom{n}{2}$ pairs of surfaces, times $m^2$ pairings of linear pieces. Thus, we can apply \Cref{thm:SE} with $k=2(\gamma+1)$ and union bound over $M\in O(n^2 m^2)$ choices by adjusting the failure probability to $\delta'=\delta/M$. Then, we obtain an embedding dimension of $t\in O(\eps^{-2}(k + \log(M/\delta)))\subseteq O(\eps^{-2}(\gamma + \log(nm/\delta)))$. 
    
    Let $A\in\R^{d\times N}$ be the matrix whose $N$ columns comprise all vertices of all piecewise linear surfaces $\sigma_1,\ldots,\sigma_n$. The above argumentation yields the existence of a suitable collection of $2(\gamma+1)$-sparse vectors $\Xi_{2(\gamma+1)}\subset \Psi_{2(\gamma+1)}$ and a linear mapping $S:\R^{d}\rightarrow\R^t$ with $t\in O(\eps^{-2}(\gamma + \log(nm/\delta)))$ such that $\forall x\in\Xi_{2(\gamma+1)}\colon (1-\eps) \|Ax\|_2 \leq \|SAx\|_2 \leq (1+\eps) \|Ax\|_2 $.
    This finally implies that
    \[
        (1-\varepsilon)\|p-q\|_2\leq \|S(p)-S(q)\|_2\leq(1+\varepsilon)\|p-q\|_2 
    \]
    holds for every pair of points $p\in\sigma_i,q\in\sigma_j$ on every pair of piecewise linear surfaces $\sigma_i,\sigma_j$, for all $1\leq i<j\leq n$, which is what we wanted to show.

    Now, to show that $S$ preserves the generalized distance measure between any given pair of piecewise linear surfaces $\sigma_i$ and $\sigma_j$, let $\{(\alpha_l^*,\beta_l^*)\}_{l\in\mathbb{N}}$ and $\{(\alpha_l^S,\beta_l^S)\}_{l\in\mathbb{N}}$ be sequences of homeomorphism pairs whose limits determine the values of $d(\sigma_i,\sigma_j)$ and $d(S(\sigma_i),S(\sigma_j))$ respectively. Then,
    \begin{align*}
        d(S(\sigma_i),S(\sigma_j)) &=\inf_{(\alpha,\beta)\in\mcT}\left(\int_{[0,1]^\gamma}\|S(\sigma_i(\alpha(r)))-S(\sigma_j(\beta(r)))\|^q_2 \,d\mu(r)\right)^{1/q} \\
        &\le \underset{l\rightarrow \infty}{\lim} \left(\int_{[0,1]^\gamma}\|S(\sigma_i(\alpha_l^*(r)))-S(\sigma_j(\beta_l^*(r)))\|^q_2 \,d\mu(r)\right)^{1/q} \\
        &\le \underset{l\rightarrow \infty}{\lim} \left(\int_{[0,1]^\gamma}(1+\eps)^q\,\|\sigma_i(\alpha_l^*(r))-\sigma_j(\beta_l^*(r))\|^q_2 \,d\mu(r)\right)^{1/q} \\
        &= (1+\eps)\,\inf_{(\alpha,\beta)\in\mcT}\left(\int_{[0,1]^\gamma}\|\sigma_i(\alpha(r))-\sigma_j(\beta(r))\|^q_2 \,d\mu(r)\right)^{1/q} \\
        &= (1+\eps)\;d(\sigma_i,\sigma_j),
    \end{align*}
    which proves the dilation bound. Similarly, 
    \begin{align*}
        d(S(\sigma_i),S(\sigma_j))
        &=\inf_{(\alpha,\beta)\in\mcT}\left(\int_{[0,1]^\gamma}\|S(\sigma_i(\alpha(r)))-S(\sigma_j(\beta(r)))\|^q_2 \,d\mu(r)\right)^{1/q} \\
        &= \underset{l\rightarrow \infty}{\lim} \left(\int_{[0,1]^\gamma}\|S(\sigma_i(\alpha_l^S(r)))-S(\sigma_j(\beta_l^S(r)))\|^q_2 \,d\mu(r)\right)^{1/q} \\
        &\ge \underset{l\rightarrow \infty}{\lim} \left(\int_{[0,1]^\gamma}(1-\eps)^q\,\|\sigma_i(\alpha_l^S(r))-\sigma_j(\beta_l^S(r))\|^q_2 \,d\mu(r)\right)^{1/q} \\
        &\geq (1-\eps)\,\inf_{(\alpha,\beta)\in\mcT}\left(\int_{[0,1]^\gamma}\|\sigma_i(\alpha(r))-\sigma_j(\beta(r))\|^q_2 \,d\mu(r)\right)^{1/q} \\
        &= (1-\eps)\;d(\sigma_i,\sigma_j),
    \end{align*}
    which proves the contraction bound. This finishes the proof.
\end{proof}

\subsection{Discrete traversals for piecewise linear surfaces}\label{sec:dFsurfaces}

In this subsection, we provide a notion of discrete traversals for surfaces and show that this definition is consistent with discrete traversals for curves. It thus suits the requirements for generalizing the discrete \Frechet distance, and variants of discrete DTW to surfaces.

In contrast with curves, for which there are definitions and algorithms for both the continuous and the discrete \Frechet distance, for surfaces only a continuous \Frechet distance has been defined in the literature~\cite{AltB10}. Given two piecewise linear surfaces $f,g:\R^\gamma\rightarrow \R^d$, the continuous \Frechet distance between $f$ and $g$ is defined as 
    \[  
        \inf_{h:\R^\gamma\rightarrow\R^\gamma}  \sup_{x\in\R^\gamma}  \|f(x)-g(h(x))\|_2, 
    \]
where the infimum is taken over all orientable homeomorphisms $h$ from $\R^\gamma$ to itself. This is consistent with the previous definition using reparameterizations by setting $h$ such that  $h \circ \alpha = \beta$. For simplicity of notation, we will treat $h\colon \R^d\rightarrow\R^d$ as a mapping between the ambient spaces of $f$ and $g$ instead of their parameter spaces.

Recall that in the definition of discrete \Frechet and DTW distances for curves, a minimum is taken over the set of discrete traversals. Hence, a definition for discrete \Frechet and DTW distance between surfaces could be obtained by defining a suitable notion of discrete traversals between surfaces. However, unlike curves, the vertices on a piecewise linear surface do not have a natural ordering that can be used to define a discrete traversal as a sequence. Instead, we define discrete traversals starting from an orientable homeomorphism.

We assume that each surface comes with a fixed `starting' point, say $u^*$ for $f$ and $v^*$ for $g$ and restrict ourselves to homeomorphisms $h$ satisfying $h(u^*)=v^*$. Let $u_1,\ldots,u_{m'}\in\R^d$ and $v_1,\ldots,v_{m''}\in\R^d$ be the sets of vertices of the surfaces $f$ and $g$, respectively. 
A discrete traversal $T_h$ of the surfaces is a set that consists of pairs of vertices, induced by $h$ in the following way. For each vertex $u_i$ of $f$, consider the Voronoi cell $V(u_i)$ of $u_i$ within $f$. Similarly, consider the Voronoi cells $V(v_j)$ for the vertices $v_j$ of $g$. We add the pair $(u_i,v_j)$ to the traversal $T_h$ (without duplicates) if and only if $h(u_i)\in V(v_j)$ or $h^{-1}(v_j)\in V(u_i)$, where ties lying on common boundaries of two or more distinct Voronoi cells are broken in lexicographic order.
We prove that a traversal constructed in this way is consistent with the standard definition of a discrete traversal in the special case of polygonal curves.

\begin{proposition}\label{prop:discretetraversal}
    Let $f,g:[0,1]\rightarrow\R^d$ be polygonal curves, let $h$ be an orientable homeomorphism between $f$ and $g$ such that $h(f(0))=g(0)$ and let $T_h$ be the discrete traversal induced by $h$ as defined above. Then, $T_h$ is a discrete traversal between curves $f$ and $g$. {Additionally, for every discrete traversal $T$, there exists an orientable homeomorphism $h$ such that $T=T_h$.}
\end{proposition}

\begin{proof}
    \textbf{Consistency } 
    Let $f(s_1)=f(0)=u_1,\ldots,f(s_{m'})=f(1)=u_{m'}$ and $g(t_1)=g(0)=v_1,\ldots,g(t_{m''})=g(1)=v_{m''}$ be the vertices of $f$ and $g$, respectively. We need to show that 
    
    \begin{enumerate}
        \item $(u_1,v_1)\in T_h \wedge (u_{m'},v_{m''})\in T_h$,
        \item $\forall i<m',j<m''\colon (u_i,v_j)\in T_h\Rightarrow (u_i,v_{j+1})\in T_h\vee(u_{i+1},v_j)\in T_h\vee (u_{i+1},v_{j+1})\in T_h$,
        \item $(u_i,v_{j+1})\in T_h \Rightarrow (u_{i+1},v_j)\notin T_h$.
    \end{enumerate}
       
    For the first property, we have by assumption, $h(u_1)=v_1$, so in particular, $h(u_1)\in V(v_1)$, which means that $(u_1,v_1)\in T_h$ holds by definition. Now, because $h$ is a homeomorphism, which maps the boundary $\partial f$ of $f$ bijectively to the boundary $\partial g$ of $g$, and since $\partial f=\{u_1,u_{m'} \}$ and $\partial g=\{ v_1,v_{m''} \}$, it follows that $h(u_{m'})=v_{m''}$, so $(u_{m'},v_{m''})\in T_h$.
   
    To prove the second property, assume that $(u_i,v_j)\in T_h$. This means that $h(u_i)\in V(v_j)$ or $h^{-1}(v_j)\in V(u_i)$. Without loss of generality, assume that $h(u_i)\in V(v_j)$. Suppose, for the sake of a contradiction, that $(u_i,v_{j+1})\not\in T_h$ and $(u_{i+1},v_j)\not\in T_h$ and $(u_{i+1},v_{j+1})\not\in T_h$. This means that $h(u_{i+1})\not\in V(v_{j})\cup V(v_{j+1})$ and $h^{-1}(v_{j+1})\not\in V(u_i)\cup V(u_{i+1})$. From $h(u_i)\in V(v_j)$ and the fact that $h$ is a homeomorphism, it follows that $h(u_{i+1})\in V(v_k)$ for some $k\geq j+2$. This implies that $h^{-1}(v_\ell)\in V(u_i)\cup V(u_{i+1})$ for all $\ell\in\{j+1,\ldots,k-1\}$. In particular, $h^{-1}(v_{j+1})\in V(u_i)\cup V(u_{i+1})$, which contradicts the statement before.

    For the third property, let $(u_i,v_{j+1})\in T_h$. This implies that either $h(u_i)\in V(v_{j+1})$ or $h^{-1}(v_{j+1})\in V(u_i)$. Then there exist $r,s\in [0,1]$ parameterizing $p=f(r), q=g(s)$ such that $h(p)=q$, and moreover $r\leq \bar r \in [0,1]$, $s\geq\bar s \in [0,1]$, where $\bar r$ parameterizes the boundary between $V(u_i)$ and $V(u_{i+1})$, and $\bar s$ parameterizes the boundary between $V(v_j)$ and $V(v_{j+1})$, respectively. Assume for the sake of contradiction that $(u_{i+1},v_j)\in T_h$. Then we have analogously $h(p')=q'$ with $p'=f(r'), q'=g(s')$ parameterized by $r' \geq \bar r$, and  $s' \leq \bar s$. This, however, violates the continuity of either $h$ or $h^{-1}$, or either $p=p'$ or $q=q'$ hold (if they lie exactly on the boundary). Any case contradicts the fact that $h$ is a homeomorphism.

    \subparagraph{Existence of a consistent homeomorphism}
    Let us now turn to showing that every traversal $T$ can be realized as $T=T_h$ for some homeomorphism $h$ through our transformation. Consider a fixed traversal $T$ and for each vertex $w$ of either of the two curves, let $\Gamma(w)=\{p\mid (w,p)\in T \vee (p,w)\in T \}$ be its neighborhood in $T$. The idea of our construction is to map the entire neighborhood of a vertex $w$ to its Voronoi cell $V(w)$ in the order induced by $T$. First, let $h(u_1)=v_1$, and $h(u_{m'})=v_{m''}$. Note that since there are no crossing edges by Item 3 above, only one of $\Gamma(u_1), \Gamma(v_1)$ can contain more than one element and the two cases are symmetric. W.l.o.g. let $\Gamma(u_1)$ contain $n_1>1$ elements $\{v_1,v_2,\ldots, v_{n_1}\}$, then we identify points $p_2,\ldots,p_{n_1}\in V(u_1)$ that appear after $u_1$ in the order given by the indices and let $h^{-1}(v_{i})=p_i$ for $i\in \{2,\ldots,n_1\}$. In case the next element in $T$ is $(u_2,v_{n_1+1})$, we can `remove' all preceding vertices and start the process over. Otherwise, the next element is $(u_2,v_{n_1})$ and we `remove' preceding vertices. We distinguish two sub-cases: if $|\Gamma(v_{n_1})|\geq 1$, then we map all  vertices to points in $V(v_{n_1})$ after $v_{n_1}$ using $h$ in the same way described above. Otherwise $|\Gamma(u_2)|\geq 1$, and we map $h(u_2)=p$ for a point $p\in V(v_{n_1})$ after $v_{n_1}$. All remaining vertices in $\Gamma(u_2)$ are mapped using $h^{-1}$ to points in $V(u_2)$ after $u_2$ in the same way described above. We repeat this algorithm until we reach the pair $(u_{m'},v_{m''})$.
    Finally, we complete the homeomorphism between the points mapped during this process in a continuous bijective way preserving the ordering of points in the image and pre-image. By construction, the traversal $T_h$ induced by $h$ equals $T$ again.
\end{proof}

\section{Further Considerations}
\subparagraph{Extension to $\ell_p$}
An interesting generalization of our framework is offered by the possibility of using $\ell_p$ normed Banach spaces as underlying vector spaces instead of the Euclidean space. This, however requires changing the i.i.d. Gaussian distribution of the embedding matrix to an i.i.d. $p$-stable distribution, which exists for any $p\in(0,2]$, see~\cite{Indyk2006,Dytso18}. Even more crucial, while in the case of the $\ell_2$ norm, the estimator is again the $\ell_2$ norm, this does not continue to hold for other choices of $p\in(0,2)$. Instead, we get the following guarantee from~\cite{MaiMM0SW23}
\[
    \forall x\in \Psi_k \colon (1-\eps) \|Ax\|_p \leq \|SAx\|_{\mathrm{med}} \leq (1+\eps) \|Ax\|_p\,,
\]
where for target dimension $t$, $\|y\|_{\mathrm{med}}\coloneqq \mathrm{median}\{|y_i|\mid i\in [t]\}$.

Any algorithm that can be implemented to query only the $\ell_p$ distances of point pairs $p,q$ located on different curves (or surfaces) will work unchanged up to $(1\pm\eps)$ errors using median queries. Thus an $\ell_p$ extension can benefit from our framework. However, algorithms that perform more complicated operations in the original dimension may need sophisticated modifications to work with the median estimator in the reduced dimension.

\subparagraph{Algorithmic considerations}
The continuous \Frechet and Hausdorff distances have efficient algorithms, see \cite{AltG95} and descendants thereof. Also discrete variants of \Frechet, DTW, Hausdorff, etc. can be computed efficiently and with linear dependencies on the dimension \cite{EiterM94}. These algorithms directly benefit from our framework by reducing the dimension of the input curves and running the algorithms unchanged on the low dimensional embedded curves.

As we have outlined in the previous paragraph, $\ell_p$ extensions may require changes to the algorithm. For continuous DTW the situation is even worse, since algorithms for these distance are only known in the $O(1)$-dimensional case \cite{BrankovicBKNPW20,BuchinNW22} (where dimension reduction makes no sense). We face similar limitations in the context of dissimilarity measures between surfaces: efficient algorithms are only known for simple polygons in $O(1)$-dimensions~\cite{BuchinBW06}. The \Frechet distance for general surfaces is only known to be semi-computable~\cite{AltB10}. However, for some of these issues, relaxations such as a \emph{weak} \Frechet distance or practical heuristics are available and benefit from dimension reduction while preserving their original performance up to $(1\pm\eps)$ errors.
Even in cases where no algorithms are known, our framework provides dimension reduction guarantees that will continue to hold for future developments. For instance, one intriguing open direction would be to extend the algorithm for the \Frechet distance between simple polygons~\cite{BuchinBW06} (or more general classes of surfaces) to arbitrary $d$-dimensional space. Our framework then applies directly to remove any dependence on $d$.

\section{Conclusion}\label{sec:conclusion}
We have introduced a widely general framework for dimension reduction for polygonal curves and piecewise linear surfaces, and discussed popular special cases as well as further extensions. We hope that our framework and methods will motivate and contribute to future developments and ultimately lead to dimension-reduced or dimension-independent coresets, dimension-free learning theoretical bounds, and finally improved algorithms for important proximity, classification and clustering problems for curves and surfaces.
A first step in exploiting random projections with collections of low-dimensional subspaces has been made recently in \cite{munteanu2026terminal}. The authors designed the first \emph{terminal embeddings} for preserving the \Frechet distance of input curves to arbitrary curves in the ambient space. This also enabled the first dimension-free and near-optimal coresets for $k$-clustering under the \Frechet distance. Extensions to other distance measures and higher dimensional linear objects are intriguing directions for further research.

\bibliography{references}

\clearpage
\appendix
\section{Equivalence of traversal definitions}\label[appendix]{app:traversal_equiv}
We prove that the usual definition of traversals as a sequence of vertex pairs is equivalent to the definition using unordered sets, that is used in our paper.
\begin{lemma}
    Let $\sigma=(u_1,\ldots,u_{m'}),\tau=(v_1,\ldots,v_{m''})$ be polygonal curves. For every sequence of vertex pairs $T=\big((u,v)_\ell\big)_{\ell\in[L]}$, where $u\in\sigma, v\in\tau$ and $L$ is the length of the sequence, consider the corresponding set $S=\{(u,v)_\ell\mid \ell\in[L]\}$. Then $T$ satisfies
    
    \begin{enumerate}[A.]
        \item $T$ starts with $(u_1,v_1)$ and ends with $(u_{m'},v_{m''})$,
        \item $\forall i<m',j<m''\colon$ the successor of $(u_i,v_j)$ in $T$ is exactly one element \\
        $(u_{i'},v_{j'}) \in\{ (u_{i+1},v_j),(u_i,v_{j+1}), (u_{i+1},v_{j+1})\},$
    \end{enumerate}

    \noindent
    if and only if $S$ satisfies
    
    \begin{enumerate}
        \item $(u_1,v_1)\in S \wedge (u_{m'},v_{m''})\in S$,
        \item $\forall i<m',j<m''\colon (u_i,v_j)\in S\Rightarrow (u_i,v_{j+1})\in S\vee(u_{i+1},v_j)\in S\vee (u_{i+1},v_{j+1})\in S$,
        \item $(u_i,v_{j+1})\in S \Rightarrow (u_{i+1},v_j)\notin S$.
    \end{enumerate}
\end{lemma}
\begin{proof}
We split the proof into several items:
    \begin{description}
        \item[A. $\Leftrightarrow$ 1.] ~\\
        If $T$ starts with $(u_1,v_1)$ and ends with $(u_{m'},v_{m''})$, then clearly $(u_1,v_1)\in S$ and $(u_{m'},v_{m''})\in S$. Reversely, if $(u_1,v_1)\in S$ and $(u_{m'},v_{m''})\in S$, we can put them to be the first, respectively the last element in $T$.
        \item[B. $\Rightarrow$ 2.] ~\\
        Let $(u_i,v_j)\in S$, then it is also an element of $T$ and has one element $(u_{i'},v_{j'}) \in \{ (u_{i+1},v_j),(u_i,v_{j+1}), (u_{i+1},v_{j+1})\}$ as its direct successor in $T$. Thus $(u_{i'},v_{j'})\in S$. 
        \item[B. $\Rightarrow$ 3.] ~\\
        Let $(u_i,v_{j+1})\in S$. Then it is also in $T$ and all its (not necessarily direct) successors $(u_{i'},v_{j'})$ in $T$ satisfy $i'\geq i$ and $j'\geq j+1$. So $(u_{i+1},v_j)$ can never appear after $(u_i,v_{j+1})$ in $T$. Similarly, all predecessors $(u_{i'},v_{j'})$ in $T$ satisfy $i'\leq i$ and $j'\leq j+1$. So $(u_{i+1},v_j)$ can never appear before $(u_i,v_{j+1})$ in $T$. It follows that $(u_{i+1},v_j)\notin S$.
        \item[2. $\wedge$ 3. $\Rightarrow$ B.] ~\\
        By Item 2 we have that at least one of the admissible successors of $(u_i,v_j)\in S$ (w.r.t. $T$) must be present in $S$. If it is exactly one, then we are done. Otherwise, by Item 3, we have that only one element $(u_{i'},v_{j'}) \in \{(u_i,v_{j+1}), (u_{+1},v_{j})\}$ can be in $S$ in addition to $(u_{i+1},v_{j+1})\in S$. We can thus construct $T$ such that $T=(\ldots,(u_i,v_j),(u_{i'},v_{j'}),(u_{i+1},v_{j+1}),\ldots)$, which obeys the rules of B. 
    \end{description}    
    
    \noindent
    This finishes our proof.
\end{proof}
\end{document}